\documentclass[12pt,a4paper]{article}

\usepackage[margin=2cm]{geometry}
\usepackage[hidelinks]{hyperref}
\usepackage{graphicx}
\usepackage{amssymb, amsmath, amsfonts, bm, mathrsfs,amsthm}
\usepackage{cleveref}

\newtheorem{corollary}{\bf Corollary}[section]
\newtheorem{lemma}{\bf Lemma}[section]
\newtheorem{proposition}{\bf Proposition}[section]
\newtheorem{definition}{\bf Definition}[section]
\newtheorem{remark}{\bf Remark}[section]

\renewcommand{\O}[0]{\mathcal{O}}
\newcommand{\R}[0]{\mathbb{R}}
\newcommand{\eps}[0]{\varepsilon}

\newcommand{\grad}[0]{\nabla}
\newcommand{\lapof}[0]{\nabla^2}
\newcommand{\divof}[0]{\nabla \cdot}
\newcommand{\abs}[1]{\left|#1\right|}
\newcommand{\br}[1]{\left(#1\right)}
\newcommand{\sbr}[1]{\left[#1\right]}
\newcommand{\pd}[2]{\frac{\partial #1}{\partial #2}}

\renewcommand{\O}[0]{\mathcal{O}}
\newcommand{\nhat}[0]{\widehat{n}}
\newcommand{\that}[0]{\widehat{t}}
\renewcommand{\hat}[0]{\widehat}
\newcommand{\Kbar}[0]{\overline{K}}
\newcommand{\Kabs}[0]{\abs{K}}
\newcommand{\Khat}[0]{\widehat{K}}

\newcommand{\Jhat}[0]{\widehat{J}}
\newcommand{\sgrad}[0]{\nabla_{\!  \bot}}
\newcommand{\grads}[0]{\nabla^\bot}
\newcommand{\sdiv}[0]{\nabla_\bot \cdot}
\newcommand{\divs}[0]{\nabla^\bot \cdot}
\newcommand{\slap}[0]{\Delta_\bot}
\newcommand{\ds}[0]{\partial_\sigma}

\newcommand{\cs}[0]{\operatorname{c}}
\newcommand{\sn}[0]{\operatorname{s}}
\newcommand{\A}[0]{\mathcal{A}}
\newcommand{\B}[0]{\mathcal{B}}
\newcommand{\C}[0]{\mathcal{C}}

\allowdisplaybreaks

\begin{document}

\title{Orthogonal signed-distance coordinates and vector calculus near evolving curves and surfaces}

\author{Eric W. Hester\thanks{Department of Mathematics, The University of California Los Angeles, Los Angeles, 90095, CA, USA}
\and
Geoffrey M. Vasil\thanks{School of Mathematics and Maxwell Institute for Mathematical Sciences, The University of Edinburgh, Edinburgh EH9 3JZ, United Kingdom}}

\maketitle

\begin{abstract}
We provide an elementary derivation of an orthogonal coordinate system for boundary layers around evolving smooth surfaces and curves based on the signed-distance function.
We go beyond previous works on the signed-distance function and collate useful vector calculus identities for these coordinates.
These results and provided code enable consistent accounting of geometric effects in the derivation of boundary layer asymptotics for a wide range of physical systems.
\end{abstract}

\section{Introduction and Background}

This paper develops a coordinate system for understanding the vicinities of smooth evolving surfaces and lines in three dimensions.
In particular, we leverage the \emph{signed-distance} function, as applied in level-set and related methods in \cite{OsherFrontsPropagatingCurvaturedependent1988, OsherLevelSetMethods2001, OsherLevelsetmethods2003, DeckelnickComputationgeometricpartial2005}.
We believe we are the first to provide a complete elementary derivation of full vector calculus of this orthogonal coordinate system, which (e.g.) is essential to describing the vector Laplacian or advective derivative commonly arising in fluid mechanics. 


Beginning with surfaces, Euler's investigations of principal curvatures \cite{EulerRecherchescourburesurfaces1767} motivated a tremendous number of problems in mathematics, including Gauss' classic Theorema Egregium \cite{GaussDisquisitionesgeneralescirca1828},
Christoffel and his eponymous symbols \cite{christoffelUeberTransformationHomogenen1869}, 
Ricci and Levi-Civita's tensor calculus \cite{RicciMethodescalculdifferentiel1900}, 
Darboux's trihedron \cite{DarbouxLeconstheoriegenerale2001}, 
and Cartan's method of moving frames \cite{Cartandeformationhypersurfacesdans1917} and differential forms \cite{Cartancertainesexpressionsdifferentielles1899}.
Carmo's textbook \cite{CarmoDifferentialgeometrycurves1976},
or Spivak's monumental five-volume work \cite{SpivakComprehensiveIntroductionDifferential1999} both provide good introductions to this vast field, including some of the histories.
After understanding the interface, researchers examined neighbourhoods of surfaces, with Weyl's seminal paper,  \cite{WeylVolumeTubes1939} calculating volumes of tubular neighbourhoods for statistical applications \cite{HotellingTubesSpheresnSpaces1939}.
The importance and utility of these neighbourhoods spawned much subsequent research in their application to the analysis of partial differential equations \cite{FedererCurvaturemeasures1959, Serrinproblemdirichletquasilinear1969, GrayTubes2004, GigaSurfaceevolutionequations2006, GilbargEllipticPartialDifferential2015}.

Unfortunately, these powerful tools are also quite complex.
The advanced state of modern differential geometry often inhibits physicists, engineers and newcomers from leveraging this understanding for tangible benefit.
Many practitioners are often more familiar with Gibbs' vector calculus \cite{WilsonVectorAnalysisTextbook1901} and straightforward orthogonal coordinate transforms \cite{BatchelorIntroductionFluidDynamics2000}.
Notably, Prandtl founded fluid boundary-layer theory using mathematics reasonably familiar to modern applied-maths undergraduates \cite{PrandtlVerhandlungenDrittenInternationalen1905}.
Several early papers built on this work to investigate three-dimensional boundary layers around general curved surfaces \cite{HowarthXXVBoundaryLayer1951, HayesThreedimensionalBoundaryLayer1951,
MooreThreeDimensionalBoundaryLayer1956, MagerThreeDimensionalLaminarBoundary1954,
StewartsonAsymptoticExpansionsTheory1957}, spawning yet further research 
\cite{SedneyAspectsThreeDimensionalBoundary1957, CookeBoundarylayersthree1962,EichelbrennerThreeDimensionalBoundaryLayers1973}.
However, managing the complexity of Navier-Stokes equations in curvilinear coordinates requires judicious approximations  \cite{DykePerturbationmethodsfluid1964, KevorkianPerturbationMethodsApplied1981, BenderAdvancedMathematicalMethods1999, HolmesIntroductionperturbationmethods2013}.
Moreover, the technical details can become byzantine, and mistakes can arise without a systematic bookkeeping method; with important terms dropped (e.g.~equation 5a of \cite{Greshopressureboundaryconditions1987}), or the assumption of coordinate systems which cannot exist (contradicting Dupin's theorem \cite{EichelbrennerThreeDimensionalBoundaryLayers1973}).
Only more recently has analysis of boundary layers in general geometries been made, beginning with arbitrary curves in the plane.

An early two-dimensional example of analysis of boundary layers in general geometries is \cite{EckhausBoundaryLayersLinear1972}. 
These methods were gradually extended to properties of the signed-distance function in three dimensions \cite{AmbrosioGeometricEvolutionProblems2000}.
Only very recently have applications to boundary layers for arbitrary objects begun \cite{CiarletIntroductionDifferentialGeometry2005, GieAsymptoticAnalysisStokes2010, GieBoundaryLayerAnalysis2012, GieRecentProgressesBoundary2015}.
However, even these most recent analyses lack three key features:
\begin{enumerate}
	\item A straightforward procedure to calculate arbitrary tensor calculus quantities (with some progress made in \cite{HesterImprovingAccuracyVolume2021}).
	\item Consideration of evolving geometries (as was done for melting and dissolving phase field models in \cite{HesterImprovedPhasefieldModels2020}).
	\item A description of neighbourhoods of curves (as would arise from the intersection of smooth surfaces or at contact lines for multiphase problems in three dimensions).
\end{enumerate}

On the last point, we note that the differential geometry of curves boasts a similarly esteemed history, beginning with Frenet \cite{FrenetCourbesDoubleCourbure1852} and Serret's \cite{SerretQuelquesFormulesRelatives1851} derivations of the Frenet-Serret formulas for describing geometric invariants of curves.
But (as we show in \cref{sec:frenet}), this convenient basis for the geometry of curves does not generate an orthogonal coordinate system (unlike the principle curvature basis for surfaces).
Instead, by offsetting the torsion of the frame with a counter-rotation, it is possible to derive \emph{relatively parallel adapted frames} \cite{BishopThereMoreOne1975}, which we build on to generate an orthogonal coordinate system in the curve neighbourhood.
This insight was also explored in \cite{DaiCompetitiveGeometricEvolution2015} for generalised Cahn-Hilliard equations and \cite{PandaDynamicsViscousFibers2006, PandaSystematicDerivationAsymptotic2008} for boundary-layer equations for slender fluid jets.
The latter work was unfairly criticised in \cite{ShikhmurzaevSpirallingLiquidJets2017}, which failed to recognise the role of rotation in ensuring orthogonality of the coordinates and instead derived more complicated analysis for non-orthogonal Frenet-Serret coordinates.
However these analyses while impressive did not provide a complete derivation suitable for arbitrary vector calculus quantities.

We aim to strike a balance between generality and simplicity.
We go beyond common treatments of the signed-distance function \cite{GigaSurfaceevolutionequations2006, GilbargEllipticPartialDifferential2015, MayostApplicationssigneddistance2014} to develop concise expressions of key differential operators (divergence, gradient, curl, scalar and vector Laplacian, and Cartesian partial time derivatives) without approximation.
At the same time, we do not invoke advanced differential geometry, exploiting only the well-known vector calculus.
Our goal is concreteness. Only with all relevant terms manifest can each contribution be evaluated (e.g.) its appropriate asymptotic order in a complex boundary-layer calculation. 
Moreover, we have previously made gains by computer-algebra automation of high-order perturbation expansions \cite{HesterImprovingAccuracyVolume2021, HesterImprovedPhasefieldModels2020}.
Inputting various formulae into a computer system requires separating all features related to the signed-distance function from all other differential-geometric aspects. 

We first discuss the basic properties of the signed-distance function in \cref{sec:sdf-sdf} before revisiting standard surface differential geometry in \cref{sec:sdf-surface}.
We then reproduce previous results on the differential geometry of signed-distance coordinates \cite[sec.~14.6]{GilbargEllipticPartialDifferential2015} in \cref{sec:sdf-normal} before collating a suite of new vector calculus identities in \cref{sec:sdf-vector-calculus,sec:sdf-dt}.
We then provide asymptotic expansions of these operators valid for thin boundary layers in \cref{sec:asymptotic-surface}.
In \cref{sec:tubes}, we apply the same coordinate system to understand the vicinities of smooth evolving curves in space.
We recapitulate the Frenet-Serret frame in \cref{sec:frenet} before showing how to rotate this frame to generate orthogonal coordinates \cref{sec:sdf-coords-tube} (building on \cite{BishopThereMoreOne1975}), 
deriving new vector calculus identities in \cref{sec:vector-calculus-tube,sec:evolving-tubes},
and providing asymptotic expansions in \cref{sec:asymptotics-tube}.

\section{Signed-distance coordinates for surfaces}
\subsection{The signed-distance function}
\label{sec:sdf-sdf}
We consider a smooth manifold, $S$, embedded within three-dimensional Euclidean space.
We can locally parameterise the surface points $p(s) \in S \subset \R^3$ with surface coordinates $s = (s_1,s_2) \in \R^2$.
We relate points in space to the surface $S$ using the \emph{signed-distance}.
\begin{definition}
	The signed-distance function, $\sigma$, of a surface $S$ is the minimum distance of a point $x \in \R^3$ to a point $p \in S$, with sign according to whether it is inside ($-$) or outside ($+$) the surface.
	\begin{align}
	\sigma(x) = \pm \min_{p\in S} \| x - p\|.
	\end{align}
\end{definition}

\begin{proposition}
	The change of coordinates formula from Cartesian to signed-distance coordinates is 
    \begin{align}\label{eq:sdf}
    x = p(s) + \sigma\, \nhat(p(s)),
    \end{align}	
    where $\nhat$ is the unit normal to $S$.
    This transformation is one-to-one sufficiently close to the surface
\end{proposition}
\begin{proof}
Every smooth surface admits a smooth tubular neighbourhood of some thickness (\cite{CarmoDifferentialgeometrycurves1976} chapter 2.7).
We have a smooth minimisation problem for the signed distance within a neighbourhood of $p(s)$.
Solving the equivalent minimisation problem for $\|x-p\|^2$ we find that at a minimum
    \begin{align}
    (x-p)\cdot \partial_s p = 0.\nonumber
    \end{align}
The Jacobian of the surface mapping $\partial_s p$ returns vectors in the tangent space.
The difference vector $(x-p)$ is therefore parallel to the unit surface normal $\nhat$ with magnitude $\sigma$.
	\end{proof}

\begin{figure}[h]
	\centering
	\includegraphics[width=.7\linewidth]{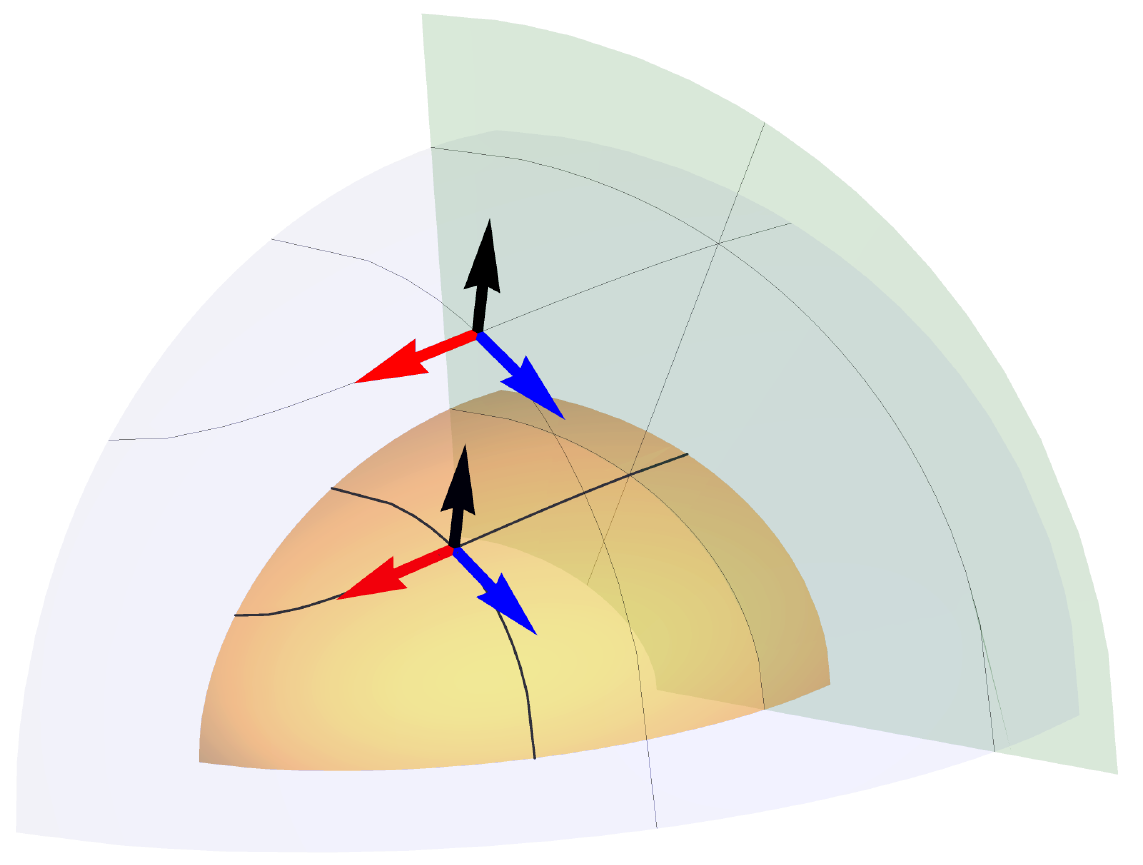} 
	\caption[signed-distance coordinate system]{The signed-distance function coordinate system. The orange surface is one octant of an ellipsoid with semimajor axes $1,\sqrt{2}$, and $2$, in orange.
	Orthogonal ellipsoidal coordinates, $s=(s_1,s_2)$, parameterise the surface with two coordinate lines in black. 
	At this surface point $p(s)$, the unit tangent vectors $\that_1,\that_2$ (red, blue) give an orthogonal basis for the tangent space.
	A point off of the boundary $x(\sigma,s)$ connects to the surface by moving a distance $\sigma$ in the normal direction $\nhat(s)$ from the closest point on the manifold $p(s)$. 
	The $\sigma = 1$ level set is shown in blue with corresponding basis vectors and surface coordinates. 
	This level set is no longer ellipsoidal, but the defined coordinate system remains orthogonal and possesses the same directions of principal curvature $\that_1,\that_2$ as the $\sigma=0$ surface.}
	\label{fig:sdf-sdf}
\end{figure}

This curvilinear coordinate system affects all vector calculus operations, which we develop in stages.
\begin{proposition}
	The partial derivatives with respect to the signed distance and surface coordinates are
    \begin{align}\label{eq:jacobian}
    \frac{\partial}{\partial \sigma} &= \nhat \cdot \nabla, &
    \frac{\partial}{\partial s_i} &= \frac{\partial p}{\partial s_i} \cdot (I + \sigma\,  \sgrad \nhat) \cdot \nabla.
    \end{align}	
\end{proposition}
\begin{proof}
By a straightforward application of the chain rule, $\frac{\partial}{\partial \sigma} = \frac{\partial x}{\partial \sigma} \cdot \grad$.
For the tangential derivatives, $\frac{\partial \nhat}{\partial s_i} = \frac{\partial p}{\partial s_i} \cdot \frac{\partial \nhat}{\partial p} = \frac{\partial p}{\partial s_i} \cdot \sgrad \nhat $, where $\sgrad$ is the surface gradient, defined in the next section.
\end{proof}

The normal derivative is straightforward, but the tangential derivatives contain scaling factors.
To understand them, we first revisit calculus {on} the manifold in \cref{sec:sdf-surface}, then probe the {normal} on the manifold in \cref{sec:sdf-normal}, and finally determine the scaling factor {off} the manifold in \cref{sec:sdf-vector-calculus}.
With these identities in hand (also see \cite[section 14.6]{GilbargEllipticPartialDifferential2015}), we then derive a suite of new vector calculus identities in \cref{sec:sdf-vector-calculus,sec:sdf-dt}.

\subsection{Surface vector calculus} 
\label{sec:sdf-surface}
In preparation for analysing the surface normal in \cref{sec:sdf-normal}, we quickly revisit some elementary differential geometry of surfaces in three dimensions.
As before, around each point, we can use coordinates $s = (s_1,s_2)$, which map to points $p(s)$ in Cartesian space.
These induce a tangent vector basis, 
\begin{align}
    {t}_i &= \frac{\partial {p}}{\partial s_i}, &
    t_i \cdot t_j &\equiv g_{ij},
\end{align}
where the metric terms $g_{ij}$ are inherited from the Euclidean dot product in space and are called the \emph{first fundamental form}.
The dot product also induces a dual covector basis
\begin{align}   
    t^i &= \nabla s_i, &
    t_i \cdot \nabla s_j &= \delta_{ij}.
\end{align}
Without loss of generality, we can specify that these coordinates are also orthogonal, and the tangent vector basis and cotangent vector basis are parallel, which allows us to define an orthonormal local vector basis
    \begin{align}
    t_i &= \frac{\partial p}{\partial s_i}, &
    \that_i &= \frac{t_i}{|t_i|}, &
    \nabla s_i &= \frac{\that_i}{|t_i|}.
    \end{align}

Introducing a surface orthogonal coordinate system expedites many of the intermediate stages of calculations, but we express all eventual results in terms of general surface covariant derivatives.  
Orthogonal coordinate systems can experience singularities around umbilical points---points of equal curvature \cite{BerryUmbilicPointsGaussian1977}.
However, these points are isolated or are otherwise subsets of the sphere or plane.
In all cases, it is possible to define an orthogonal coordinate patch around the point.
And coordinate-invariant geometric operators (gradient etc.) are well defined at all points.

For orthogonal surface coordinates, the surface gradient is 
\begin{align}
\sgrad = \nabla s_1  \frac{\partial}{\partial s_1} + \nabla s_2  \frac{\partial}{\partial s_2} = \frac{\that_1}{|t_1|} \frac{\partial}{\partial s_1} + \frac{\that_2}{|t_2|} \frac{\partial}{\partial s_2} = \that_1 \nabla_1 + \that_2 \nabla_2.
    \end{align}
The first equality follows immediately from the chain rule, the second equality follows from the orthogonality of the tangent vectors, and the final equality defines useful rescaled derivative operators, normalised to the arc length of curves through the surface.

The area measure of the surface is  
\begin{align}
	dA = |t_1||t_2| \, ds_1 \, ds_2.
\end{align}
In general, the area measure is the determinant of the surface Jacobian, which is equivalent to the norm of the cross product of the tangent vectors $dA = |t_1 \times t_2|\, ds_1\, ds_2$.
For orthogonal coordinates we have $|t_1 \times t_2| = |t_1||t_2|$.

The surface divergence of a vector $u_\bot = u_1 \that_1 + u_2 \that_2$ is 
\begin{align}
\sdiv u_\bot = \frac{1}{|t_1||t_2|}\br{\pd{}{s_1}(|t_2| u_1) + \pd{}{s_2}(|t_1| u_2)}.
\end{align}
The surface divergence is the adjoint of the surface gradient.
For a smooth function $f$ that vanishes on the boundary of a smooth set $U$, 
\begin{align*}
\int_U u_\bot \cdot \sgrad f \, dA
	&= \int_U \left[ \frac{u_1}{|t_1|} \frac{\partial f}{\partial s_1} + \frac{u_2}{|t_2|} \frac{\partial f}{\partial s_2} \right] |t_1||t_2| \, ds_1 \, ds_2,\\
	&=-\int_U f \frac{1}{|t_1||t_2|}\br{\pd{}{s_1}(|t_2| u_1) + \pd{}{s_2}(|t_1| u_2)} |t_1||t_2| \, ds_1 \, ds_2,\\
	&=- \int_U f \, \sgrad \cdot u_\bot \, dA.
	\end{align*}

\subsection{Surface normal} 
\label{sec:sdf-normal}
We now examine standard properties of the normal to the manifold at $\sigma=0$, defined by $\nhat = \that_1 \times \that_2$, also known as the \emph{Gauss map}.
The key reason for the utility of the signed distance is the following fact.
\begin{proposition}	
	The gradient of the signed-distance function is the unit normal of the closest point on the surface
	\begin{align}\label{eq:gradsigma}
    \nabla \sigma = \nhat.
    \end{align}
Moreover, this result lifts off the surface into some neighbourhood of $\R^{3}$.
\end{proposition}

\begin{proof}
	From \cref{eq:sdf}, we can express the signed distance as $\sigma = (x - p)\cdot\nhat$.
	We can then take the gradient of this relation to find
	\begin{align*}
    \nabla \sigma 
    	&= (\nabla x) \cdot \nhat - (\nabla p) \cdot \nhat + (\nabla \nhat) \cdot (x-p),\\
    	&= I \cdot \nhat - (\nabla p) \cdot \nhat + \sigma (\nabla \nhat) \cdot \nhat,\\
    	&= \nhat.
    \end{align*}
    We emphasise that the contraction does not occur on the gradient and express this by enclosing it in parentheses and keeping it on the left of the dot product.
	The second term of the second equation is zero, as the gradient of the mapping must lie within the tangent space of the surface.
	That the final term of the second equation is zero follows from taking the gradient of the (unit) magnitude of the normal vector $\nabla (\nhat \cdot \nhat) = 2 (\nabla \nhat) \cdot \nhat = 0$.
\end{proof}    

This implies that the normal is its own dual vector and is the key reason for the utility of the signed-distance function coordinate system.
From this follows several important results.
\begin{corollary}
	The gradient of the normal is the Hessian of the signed-distance function, which is a necessarily symmetric 2-tensor,
	\begin{align}
    \nabla \nhat &= \nabla \nabla \sigma.
    \end{align}
\end{corollary}
\begin{proof}
	The identity follows by taking the gradient of \cref{eq:gradsigma}.
	The symmetry of the Hessian is manifest in Cartesian coordinates.
\end{proof}

\begin{corollary}
	The normal vector is a zero eigenvector of the normal gradient.
	Hence, the normal gradient of the normal vector is zero.
	\begin{align}
	\partial_\sigma \nhat = 0.
	\end{align}	
 This is, the normal trivially lifts off the surface into $\R^{3}$. 
\end{corollary}
\begin{proof}
	The previously noted identity $(\nabla \nhat)\cdot \nhat$ implies that $\nhat$ is a zero eigenvector of $\nabla \nhat$.
	The normal gradient's symmetry implies the normal vector's normal derivative is also zero $(\nabla \nhat) \cdot \nhat = (\nhat \cdot \nabla) \nhat = 0$.
\end{proof}

\begin{corollary}
	The remaining eigenvectors of the gradient of the unit normal lie within the tangent space of the surface and are orthogonal.	
\end{corollary}

\begin{proof}
	The existence of orthogonal eigenvectors and eigenvalues follows from the symmetry of the normal gradient and the spectral theorem for symmetric tensors.
	The remaining eigenvectors must be orthogonal to the $\nhat$, the zero eigenvector, and so lie within the tangent space.
\end{proof}

\begin{definition}
	The surface gradient of the normal vector is the shape tensor $-K$.
	The eigenvectors of the shape tensor are the principal directions of curvature, and their eigenvalues are the (negative) principal curvatures.
    \begin{align}
    \sgrad \nhat = - K = - \kappa_1 \, \that_1 \otimes \that_1 - \kappa_2\,  \that_2 \otimes \that_2.
    \end{align}
\end{definition}
	
The orthonormal basis vectors $(\nhat, \that_1, \that_2)$ induced by this choice of coordinates is an example of a Darboux frame.
Without loss of generality (except at isolated umbilic points), we can locally let the surface coordinates lie parallel to these eigenvectors, meaning the unit tangent vectors are the principal curvature directions.

\subsection{Vector calculus in the boundary region}
\label{sec:sdf-vector-calculus}
We can now calculate the gradient throughout all space.
\begin{proposition}
	The gradient in signed-distance coordinates is
    \begin{align}
    \nabla &= \nhat\, \ds + J^{-1} \cdot \nabla_\bot, \quad \text{where} \quad     J = I - \sigma\, K.
    \end{align}	
\end{proposition}
\begin{proof}
	This is a simple consequence of inverting \cref{eq:jacobian}.
	If our surface coordinates are tangent to the principal curvature directions,
	\begin{align*}
    \frac{\partial x}{\partial s_i} &= \frac{\partial p}{\partial s_i} \cdot (I - \sigma K) \cdot \nabla x = (1-\sigma \kappa_i)|t_i| \that_i.
    \end{align*}
    Hence the principal tangent basis remains orthogonal off the surface (each scaled by $1-\sigma \kappa_i$).
    The dual basis is, therefore, orthogonal and scaled by the inverse factor
	\begin{align}    
	\nabla s_i &= \frac{1}{1-\sigma \kappa_i} \frac{\that_i}{|t_i|}.\nonumber
    \end{align}

    We then apply the chain rule using the dual basis, writing the curvature factors using $J^{-1}$.
\end{proof}

	\begin{corollary}
	The Hessian of the signed distance is given everywhere by
	\begin{align}
    \nabla \nabla \sigma &= \nabla \nhat = - J^{-1} K.
    \end{align}
	\end{corollary}
	
	\begin{corollary}
	Principal directions of curvature do not vary in $\sigma$.
	\end{corollary}

The gradient allows us to calculate all remaining vector calculus quantities.
The basis vectors $(\nhat, \that_1, \that_2)$ are independent of $\sigma$, giving us a basis for vector fields $u$ away from the boundary,
    \begin{align}
    u = u_\sigma \, \nhat + u_\bot, \quad \text{where} \quad u_\bot = u_1\, \that_1 + u_2 \, \that_2.
    \end{align}

	\begin{remark}
	We note that the determinant of $J$ is 
    \begin{align}
    |J| &= 1 - 2\sigma \Kbar + \sigma^2 \Kabs, \qquad \text{where } \qquad \Kbar = \kappa_1 + \kappa_2  \quad \text{ and } \quad \Kabs = \kappa_1 \kappa_2.
    \end{align}
	The mean and Gaussian curvature result from the respective trace and determinant, $\Kbar = \text{tr}(K)/2$, $\Kabs = \det(K)$.	
	\end{remark}

	\begin{proposition}
	The full volume measure $dV$ is 
    \begin{align}
    dV = |J| \, d\sigma \, dA = |J| |t_1| |t_2| \, d\sigma \, ds_1 \, ds_2.
    \end{align}
	\end{proposition}
    \begin{proof}
    This follows simply from the triple product of the new tangent vectors,
    \begin{align*}
	dV  &= (\partial_\sigma x \cdot (\partial_{s_1} x \times \partial_{s_2} x)) \,d\sigma \, d s_1 \, d s_2, \\
		&= |t_1||t_2|(1-\sigma \kappa_1)(1-\sigma\kappa_2) \nhat \cdot (\that_1\times\that_2) \, d\sigma \, d s_1 \, d s_2,\\
		&= |J| d\sigma \, dA.
    \end{align*}
    \end{proof}
    
    \begin{proposition}
The volume divergence is 
    \begin{align}
    \divof{u} &= \frac{\ds (|J| u_\sigma)}{|J|} + \frac{\nabla_\bot \cdot (\Jhat\, u_\bot)}{|J|},
    \end{align}
where we have defined for convenience the adjugate matrices 
	\begin{align}
    \Jhat &= |J|J^{-1} = I - \sigma \Khat & \Khat &= |K| K^{-1} = \kappa_2 \, \that_1 \otimes \that_1 + \kappa_1 \, \that_2 \otimes  \that_2,
    \end{align}
which (in two dimensions) swap the principal curvatures of the shape operator.
    \end{proposition}
    \begin{proof}
    The proof is analogous to the integration by parts method for the surface divergence.
    \end{proof}

	\begin{corollary}
	The scalar Laplacian is the divergence of a scalar gradient
    \begin{align}
    \lapof f &= \frac{\ds(|J|\ds f)}{|J|} + \frac{\sdiv(\Jhat J^{-1} \sgrad f)}{|J|}.
    \end{align}
	\end{corollary}
	
	\begin{corollary}
	The curl is   
    \begin{align}
    \nabla \times u = -\nhat \frac{\sdiv(\Jhat u^\bot)}{|J|} + \Jhat^{-1} (\ds (\Jhat u^\bot) - \grads u_\sigma),
    \end{align}
	where for convenience, we define rotated quantities as
    \begin{align}
    \nabla^\bot &= \nhat \times \nabla_\bot, & u^\bot &= \nhat \times u_\bot,
    \end{align}
	which satisfies the useful identities
	\begin{align}
    \nabla^\bot \cdot \nabla_\bot &= \nabla_\bot \cdot \nabla^\bot = 0, & u^\bot \cdot u_\bot &= u_\bot \cdot u^\bot = 0,\\
    \nhat \times (J u_\bot) &= \Jhat u^\bot, & \nhat \times u^\bot &= -u_\bot.
    \end{align}
	\end{corollary}
	\begin{proof}
	This follows from the fact that the curl is the unique operator (up to sign) satisfying $\nabla \cdot \nabla \times = \nabla \times \nabla = 0$.
	\end{proof}

    \begin{corollary}
	The vector Laplacian is  
    \begin{align}
    \lapof u &= \nhat \sbr{\ds \br{\frac{\ds(|J|u_\sigma) + \sdiv(\Jhat u_\bot)}{|J|}} - \frac{\sdiv(\Jhat J^{-1}(\ds(J u_\bot) - \sgrad u_\sigma )}{|J|}}\\
    &\quad+ \Jhat^{-1} \ds\br{\Jhat J^{-1} (\ds(J u_\bot) - \sgrad u_\sigma)} + J^{-1} \sgrad\br{\frac{\ds(|J|u_\sigma)}{|J|}} + \slap u_\bot ,\nonumber
    \end{align}
	where we define the vector surface Laplace operator,
	\begin{align}
    \slap u_\bot &= J^{-1} \sgrad \br{\frac{\sdiv(\Jhat u_\bot)}{|J|}} + \Jhat^{-1} \grads\br{\frac{\divs(J u_\bot)}{|J|}}.
    \end{align}
    \end{corollary}
    \begin{proof}
    This is a straightforward application of the identity $\nabla^2 u = -\nabla \times \nabla \times u + \nabla (\nabla \cdot u)$.
    \end{proof}
    
    \begin{corollary}
	The divergence-free vector Laplacian of incompressible hydrodynamics is
    \begin{align}
    -\nabla \times \nabla \times u &= 
    \frac{\nhat}{|J|} \sbr{ - \sdiv(\Jhat J^{-1}(\ds(J u_\bot) - \sgrad u_\sigma )}\\
    &\quad + \Jhat^{-1} \ds\br{\Jhat J^{-1} (\ds(J u_\bot) - \sgrad u_\sigma)} + \Jhat^{-1} \grads\br{\frac{\divs(J u_\bot)}{|J|}}. \nonumber
    \end{align}
    \end{corollary}

	\begin{proposition}
	The gradients of the basis vectors are 
    \begin{align}
    \grad \nhat &= J^{-1} \cdot \sgrad \nhat, & \grad \that_i &= J^{-1} \cdot \sgrad \that_i.
    \end{align}
	The surface gradients are 
    \begin{align}
    \sgrad \nhat &= - \kappa_1\, \that_1 \otimes \that_1 -\kappa_2 \, \that_2 \otimes \that_2, & 
    \sgrad \that_i &= \kappa_i \, \that_i \otimes \nhat + \mathcal{R}^{j k}_i \, \that_j \otimes \that_k,
    \end{align}
	using Einstein summation convection for repeated raised and lowered indices. 
 The Ricci rotation coefficients $\mathcal{R}^{jk}_i$ are antisymmetric with two independent degrees of freedom
    \begin{align}
    \label{eq:sdf-basis-gradient}
    (\nabla_j \that_i) \cdot \that_k \equiv \mathcal{R}^{jk}_i &= - \mathcal{R}^{ji}_k, &
    \mathcal{R}^{12}_1 &= \omega_1, &
    \mathcal{R}^{21}_2 &= -\omega_2.
    \end{align}
	\end{proposition}
	\begin{proof}
    As the vector basis is independent of $\sigma$ (i.e.~$\ds \nhat = \ds \that_i = 0$), the vector gradients are rescaled surface gradients.
	The antisymmetry of the Ricci rotation coefficients follows from the fact that the basis is orthonormal
    \begin{align*}
	\nhat \cdot \that_i &= 0 \implies \grad \nhat \cdot \that_i = - \grad \that_i \cdot \nhat,&
	\that_i \cdot \that_j &= 0 \implies \grad \that_i \cdot \that_j = - \grad \that_j \cdot \that_i.
    \end{align*}
	\end{proof}
	
    \begin{corollary}
    \label{cor:surface-vector-gradient}
    The velocity vector gradient is 
    \begin{align}
    \grad u &= \nhat \otimes (\nhat \,\ds u_\sigma +  \ds u_\bot  )+ J^{-1}\left((\sgrad u_\sigma+ K u_\bot)\otimes \nhat + (\sgrad u_\bot\cdot\Pi - u_\sigma \, K )\right).
    \end{align}
    where the surface projection tensor $\Pi = I -\nhat \otimes \nhat$.
    \end{corollary}

    \begin{corollary}
	The convective derivative is
    \begin{align}
	u \cdot \grad u = \nhat \,\big(&u_\sigma \ds u_\sigma+ u_\bot \cdot J^{-1} (\sgrad u_\sigma+ K u_\bot)\big) + u_\sigma \ds u_\bot + u_\bot \cdot J^{-1} (\sgrad u_\bot - K u_\sigma).
    \end{align}
	\end{corollary}
    
    \begin{corollary}
    \label{thm:hessian}
	The Hessian is 
    \begin{align*}
    \nabla \nabla &= \nhat \otimes \ds(\nhat \ds + J^{-1} \cdot \sgrad) + J^{-1} \cdot \sgrad (\nhat \ds + J^{-1} \cdot \sgrad),\\
    &= \nhat \otimes \nhat\, \ds^2 
     + \nhat \otimes \that_1\br{\frac{\kappa_1 \nabla_1}{(1-\sigma\kappa_1)^2} + \frac{\ds \nabla_1}{1-\sigma\kappa_1}} 
     + \nhat \otimes \that_2\br{\frac{\kappa_2 \nabla_2}{(1-\sigma\kappa_2)^2} + \frac{\ds \nabla_2}{1-\sigma\kappa_2}},\\
    &+ \that_1 \otimes \that_1 \br{-\frac{\kappa_1\ds}{1-\sigma\kappa_1} + \frac{\nabla_1\nabla_1}{(1-\sigma\kappa_1)^2} - \frac{\omega_1 \nabla_2}{|J|} + \frac{\sigma\nabla_1\kappa_1\nabla_1}{(1-\sigma\kappa_1)^3}},\\
    &+ \that_1 \otimes \that_2\br{\frac{\nabla_1\nabla_2}{|J|} + \frac{\omega_1\nabla_1}{(1-\sigma\kappa_1)^2} + \frac{\sigma\nabla_1\kappa_2\nabla_2}{|J|(1-\sigma\kappa_2)}},\\
    &+ \that_1 \otimes \nhat\br{\frac{\nabla_1\ds}{1-\sigma\kappa_1}+\frac{\kappa_1\nabla_1}{(1-\sigma\kappa_1)^2}},\\
    &+ \that_2 \otimes \that_2 \br{-\frac{\kappa_2\ds}{1-\sigma\kappa_2} + \frac{\nabla_2\nabla_2}{(1-\sigma\kappa_2)^2} + \frac{\omega_2 \nabla_1}{|J|} + \frac{\sigma\nabla_2\kappa_2\nabla_2}{(1-\sigma\kappa_2)^3}},\\
    &+ \that_2 \otimes \that_1\br{\frac{\nabla_2\nabla_1}{|J|} - \frac{\omega_2\nabla_2}{(1-\sigma\kappa_2)^2} + \frac{\sigma\nabla_2\kappa_1\nabla_1}{|J|(1-\sigma\kappa_1)}},\\
    &+ \that_2 \otimes \nhat\br{\frac{\nabla_2\ds}{1-\sigma\kappa_2}+\frac{\kappa_2\nabla_2}{(1-\sigma\kappa_2)^2}}.
    \end{align*}
    \end{corollary}
    \begin{proof}
    By definition, this identity follows from composing the gradient with itself.
    \end{proof}

	\begin{corollary}
	The commutators of the derivatives are
    \begin{align}
    [\ds,\nabla_1] &= 0,\\
    [\ds,\nabla_2] &= 0,\\
    [\nabla_1,\nabla_2] &= -\frac{|J|}{(1-\sigma\kappa_1)^2}{\omega_1\nabla_1}-\frac{|J|}{(1-\sigma\kappa_2)^2}{\omega_2\nabla_2} + \sigma\br{\frac{\nabla_2\kappa_1}{1-\sigma\kappa_1}\nabla_1 - \frac{\nabla_1\kappa_2}{1-\sigma\kappa_2}\nabla_2}.
     \label{eq:sdf-commutator}
    \end{align}
	\end{corollary}
	\begin{proof}
	These identities follow from enforcing symmetry of the Hessian in corollary \ref{thm:hessian}.
	\end{proof}
	
	\begin{corollary}
	The surface divergence in terms of the rotation coefficients is
    \begin{align}
    \sgrad \cdot u_\bot = \operatorname{tr}(\sgrad u_\bot) = \nabla_1 u_1 + \nabla_2 u_2 + \omega_2 \,  u_1 - \omega_1 \, u_2.
    \end{align}
    \end{corollary}
    \begin{proof}
	This comes from the trace of the vector gradient, and the fact $\nabla_1\log|t_2| = \omega_2$ and $\nabla_2\log|t_1| = - \omega_1$, which follows from
    \begin{align*}
    \omega_2 &\equiv (\nabla_2 \that_1) \cdot \that_2,\\
    &= \frac{1}{|t_2|^2} \pd{}{s_2}\br{\frac{t_1}{|t_1|}}\cdot t_2,\\
    &= \frac{1}{|t_2|^2|t_1|}\pd{t_1}{s_2}\cdot t_2 + 0,\\
    &= \frac{1}{|t_2|^2|t_1|}\pd{t_2}{s_1}\cdot t_2, \\
    &= \frac{1}{|t_2|}\nabla_1(|t_2|\that_2)\cdot \that_2,\\
    &= \nabla_1 \log|t_2| + 0.
    \end{align*}
	The procedure is similar for $\omega_1$.
    \end{proof}
    
    \begin{corollary}
    The derivatives of the curvatures and rotation coefficients obey,
	\begin{align}
	\nabla_2 \kappa_1 - \omega_1(\kappa_1 -\kappa_2)&= 0,\\
	\nabla_1 \kappa_2 - \omega_2(\kappa_1 -\kappa_2)&= 0,\\
	\nabla_1 \omega_2 - \nabla_2 \omega_1 +\omega_1^2 + \omega_2^2 &= - \kappa_1 \kappa_2.	
	\end{align}
    \end{corollary}
    \begin{proof}
    First, directly calculate the commutator of each basis vector at the surface ($\sigma=0$) using \cref{eq:sdf-basis-gradient}.
    Then, compare with the commutators given in \cref{eq:sdf-commutator}.
    \end{proof}
    \begin{remark}
   	The first two identities are known as the Codazzi-Mainardi equations.
	The final identity is Gauss' Theorema Egregium---the Gaussian curvature can be defined using intrinsic geometric quantities.
    \end{remark}

	This completes the required vector calculus for body-centred coordinates near a boundary.

\subsection{Evolving surfaces}
\label{sec:sdf-dt}
The covered identities provide straightforward tools for describing the spatial variation of quantities near boundaries.
However, we will also consider problems in which the surface may change over time.
This requires us to augment the time derivative operators.
The full transformation between Cartesian coordinates $(x,t)$ and signed-distance coordinates $(s, \sigma, \tau)$ is 
    \begin{align}
    x &= p(s,\tau) + \sigma \, \nhat(p(s,\tau)), &
    t &= \tau.
    \end{align}
    \begin{remark}
    We emphasise that despite the equality of the time coordinates $t$ and $\tau$, the partial derivatives $\partial_t$ and $\partial_\tau$ will \emph{not} be equal in general.
    \end{remark}
    
    \begin{remark}
    We do not assume that the surface coordinates remain parallel to lines of principal curvature under evolution of the surface.
    Therefore the tangent vectors $t_i = \partial_{s_i}p$ will not in general be orthogonal.
    As before, the purpose is to derive coordinate-invariant expressions for physical quantities.
    \end{remark}
    
    We first consider the Cartesian time derivatives of the signed-distance coordinates.
    
    \begin{definition}
    \label{def:surface-velocity}
    The signed-distance partial time derivative of the surface coordinate $p(s,\tau)$ (i.e.~holding $s$ constant) is defined to be,
    \begin{align}
    \partial_\tau p(s,\tau) \ \equiv  \ v_\sigma\,  \nhat(s,\tau) + v_\bot(s,\tau) = v_\sigma \nhat + v_i t_i,
    \end{align}
    where $v_\sigma$ and $v_\bot$ are the local normal and tangential velocity of the surface.
    \end{definition}

    \begin{remark}
    In Stefan boundary problems \cite{LeBarsInterfacialConditionsPure2006,VasilDynamicBifurcationsPattern2011,HesterImprovedPhasefieldModels2020} the motion of the interface is often described by coordinates that move normal to the interface.
    However, it can be convenient to generalise to coordinates that also move tangential to the boundary, such as for material coordinates of a translating elastic membrane in fluid flow.
    \end{remark}

    \begin{lemma}
    \label{lem:surface-velocity-surface-gradient}
    The surface gradient of the tangent velocity is given by
    \begin{align}
    \sgrad v_\bot &= (\sgrad v_\bot) \cdot \Pi + (K \cdot v_\bot)\otimes \nhat,
    \end{align}
    \end{lemma}
    \begin{proof}
    Take the surface projection of \cref{cor:surface-vector-gradient} on the surface velocity $v_\bot$ at $\sigma=0$.
    \end{proof}
        
    \begin{lemma}
    \label{lem:sdf-dtaut}
    The signed-distance time derivative of the tangent vectors is 
    \begin{align}
    \partial_\tau t_i &= t_i \cdot ((\sgrad v_\sigma + K v_\bot )\otimes \nhat  - v_\sigma K + \sgrad v_\bot \cdot \Pi).
    \end{align}
    \end{lemma}
    \begin{proof}
    This is a straightforward procedure that relies on the equality of mixed partials in moving signed-distance coordinates,
    \begin{align*}
    \partial_\tau t_i 
    &= \partial_{\tau} \left(\partial_{s_i} p\right),\\
    &= \partial_{s_i}\partial_\tau p,\\
        &= \partial_{s_i}(v_\sigma \nhat + v_\bot),\\
        &= t_i \cdot \sgrad(v_\sigma \nhat + v_\bot),\\
        &= t_i \cdot ((\sgrad v_\sigma) \otimes \nhat +v_\sigma \sgrad \nhat + \sgrad v_\bot),\\
        &= t_i \cdot ((\sgrad v_\sigma)\otimes \nhat + v_\sigma \sgrad \nhat + \sgrad v_\bot \cdot \Pi + (Kv_\bot)\otimes \nhat).
    \end{align*}
    \end{proof}
    
    \begin{lemma}
    \label{lem:sdf-dtaun}
    The partial time derivative of the normal vector is given by
    \begin{align}
    \partial_\tau \nhat = -(\sgrad v_\sigma + K v_\bot).
    \end{align}
    \end{lemma}
    \begin{proof}
    This is a consequence of orthogonality and \cref{lem:sdf-dtaut},
    \begin{align*}
    \partial_\tau (\nhat\cdot\nhat) &= 0 \implies \nhat \cdot \partial_\tau \nhat = 0,\\
    \partial_\tau (\nhat\cdot t_i) &= 0 \implies t_i \cdot \partial_\tau \nhat = -\nhat \cdot \partial_\tau t_i = -t_i \cdot (\sgrad v_\sigma + K v_\bot).
    \end{align*}
        \end{proof}
        
    \begin{lemma}
    \label{lem:sdf-dts}
    The Cartesian partial time derivatives of the signed-distance coordinates are
    \begin{align}
    \label{eq:surface-dtcoords}
    \partial_t \sigma &= - v_\sigma , &
    \partial_t s_i t_i &= - v_\bot + \sigma J^{-1}\cdot \sgrad v_\sigma,
    \end{align}
    where we have used the Einstein summation convention for the surface coordinate derivatives.
    \end{lemma}	
    \begin{proof}
    The chain rule gives the Cartesian partial time derivative as
    \begin{align*}
    \partial_t &= \partial_\tau + \partial_t\sigma \partial_\sigma + \partial_t s_i \cdot \pd{}{s_i} =  \partial_\tau + \partial_t\sigma \partial_\sigma + \partial_t s_i t_i \cdot \sgrad. 
    \end{align*}
    We apply this to the Cartesian coordinates, which must necessarily give zero.
    Using \cref{lem:sdf-dtaut,lem:sdf-dtaun}, we find
    \begin{align*}
    \partial_t x &= \partial_t (p + \sigma \nhat),\\
        &= (\partial_\tau + \partial_t \sigma \partial_\sigma + \partial_t s_i t_i \cdot \sgrad)(p + \sigma \nhat),\\
        &= \partial_\tau p + \partial_t s_i t_i \cdot \sgrad p + \sigma \partial_\tau \nhat + \partial_t \sigma \nhat + \sigma \partial_t s_i t_i \cdot \sgrad \nhat,\\
        &= (v_\sigma + \partial_t \sigma) \nhat + (v_\bot + \sigma \partial_\tau \nhat + \partial_t s_i t_i \cdot (\Pi - \sigma K)),\\
        &= (v_\sigma + \partial_t \sigma) \nhat + (v_\bot - \sigma(\sgrad v_\sigma +K v_\bot)   + \partial_t s_i t_i \cdot (\Pi - \sigma K)),\\
        &= (v_\sigma + \partial_t \sigma) \nhat + ((v_\bot +\partial_t s_i t_i )\cdot (\Pi - \sigma K) - \sigma \nabla_\bot v_\sigma).
    \end{align*}
    Requiring both the normal and tangential components to be zero gives \cref{eq:surface-dtcoords}.
    \end{proof}
	
    \begin{corollary}
    The Cartesian partial time derivative in moving signed-distance coordinates is
    \begin{align}
    \partial_t &= \partial_\tau  - v_\sigma \partial_\sigma -v_\bot \cdot \sgrad + \sigma \sgrad v_\sigma \cdot J^{-1} \cdot \sgrad.
    \end{align}
    \end{corollary}	
    
    \begin{corollary}
    The Cartesian time derivative of the vector velocity is,
    \begin{align}
    \partial_t u &= \br{\partial_\tau u_\sigma - v\partial_\sigma u_\sigma + (\sigma \sgrad v \cdot J^{-1} - v_\bot )\cdot (\sgrad u_\sigma + K u_\bot)}\nhat\nonumber\\
    &+\br{\partial_\tau u_\bot - v\partial_\sigma u_\bot + (\sigma \sgrad v \cdot J^{-1} - v_\bot ) \cdot (\sgrad u_\bot - K u_\sigma)},
    \end{align}
    where
    \begin{align}
    \partial_\tau u_\bot = (\partial_\tau u_i) t_i + (u_\bot \cdot (\sgrad v_\sigma + K\cdot v_\bot)) \nhat + u_\bot \cdot (\sgrad u_\bot \cdot \Pi - v_\sigma K).\nonumber
    \end{align}
    \end{corollary}
    \begin{proof}
    This follows from considering the time derivative of the velocity components and signed-distance basis vectors.
    \end{proof}

\subsection{Application: Surface Boundary Layer Coordinates}
\label{sec:asymptotic-surface}
A chief application of the signed-distance coordinate system is to furnish a useful coordinate system for boundary layers, including possible time evolution. The coordinates are \textit{per se}, not a numerical/computational device. The overall idea is to use signed-distance coordinates as a tool for \textit{mathematical analysis} that can subsequently distil more tractable (and hence computationally feasible) model equations. We took this approach to good ends in studying (e.g.) complex ice-melting--buoyancy-driven convection studies \cite{HesterImprovedPhasefieldModels2020}, extending earlier work in diffuse-domain \cite{LiSolvingPDEsComplex2009} and phase-field methods \cite{PlappPhaseFieldModels2012}.

\begin{remark}
Signed-distance coordinates can form singularities when there are multiple closest points on the surface (e.g. \cref{fig:sdf-singularities}).
Problems happen when $J$ is non-invertible,  
\begin{align}
\sigma_{*} \ = \ \frac{\Kbar \pm \sqrt{\Kbar^{2} - |K|}}{|K|}  \ = \  \kappa_{1}^{-1}, \, \kappa_{2}^{-1}.
\end{align}
But around every smooth surface, we are guaranteed the existence of a tubular neighbourhood on which the signed-distance function remains smooth (blue region).
This is, precisely, the use case for infinitesimally thin boundary layers used in multiple scales matched asymptotics. 
\end{remark}	

\begin{remark}
Ordinary spherical/cylindrical polar coordinates are trivial special cases of signed-distance coordinates, where singularities only occur at $x = y = z = 0$ (spherical) or $x=y=0$ (cylindrical). Nevertheless, both coordinate systems accommodate globally analytic functions with special restrictions at the singularities \cite{VasilTensorCalculusPolar2016,VasilTensorCalculusSpherical2019}. 
Studying the finer algebraic and combinatorial properties of such functions near these coordinate provides much information regarding the geometry and symmetry of spheres and solid balls. Moreover, such considerations provide stable, accurate and efficient numerical algorithms, notwithstanding singularities. An analysis of functions in signed-distance coordinates would produce interesting insights into the geometry of more general surfaces.   
\end{remark}

\begin{figure}[hbt]
	\centering
    \includegraphics[width=.6\linewidth]{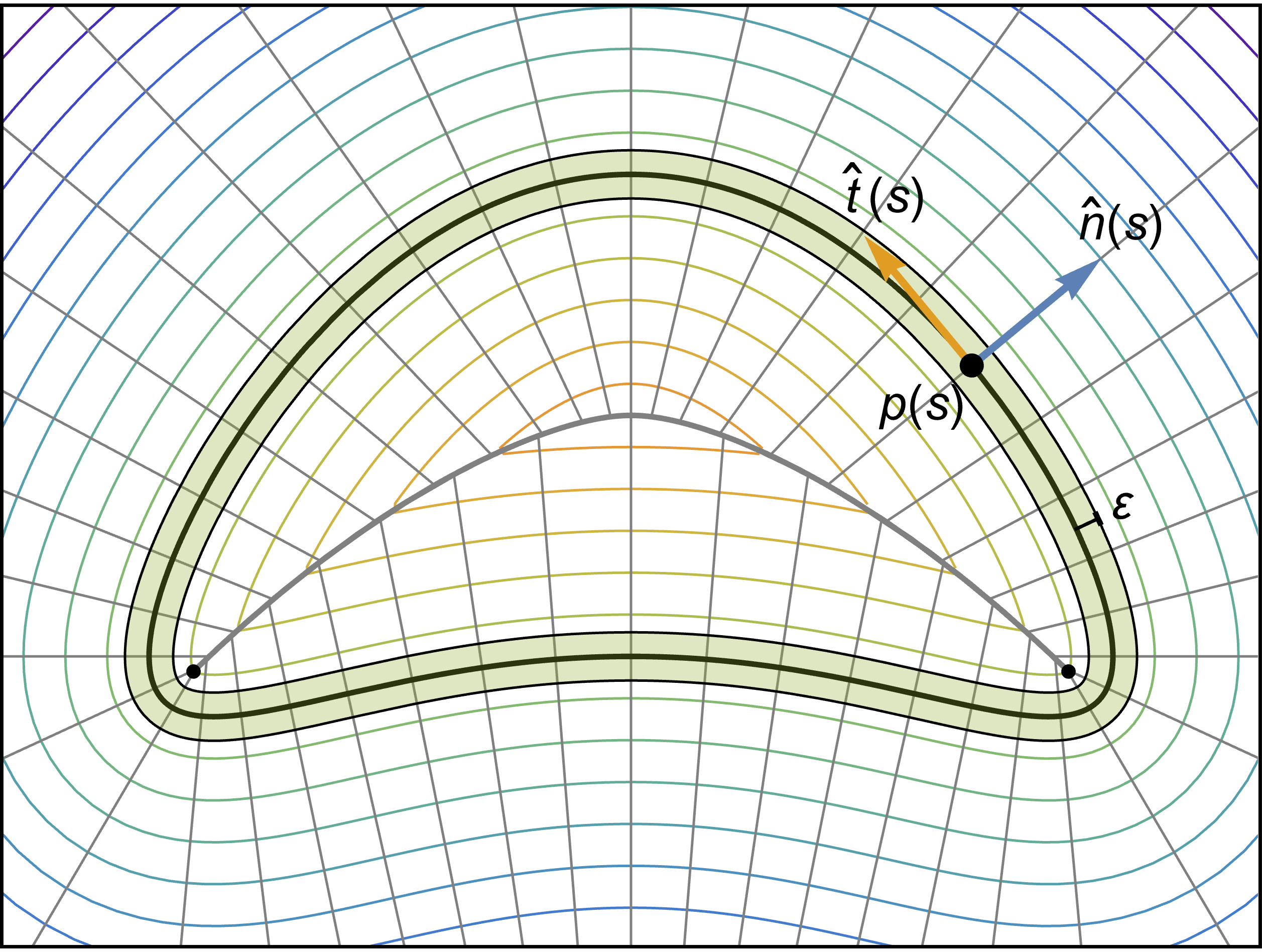}
    \caption{signed-distance coordinates around curved surface form singularities.
    They either result from lying on the curve evolute (black points) or from losing uniqueness of the closest point (grey curve).
    But every smooth surface admits a finite thickness tubular neighbourhood on which the signed-distance $\sigma$ remains smooth (green boundary region).
    Signed-distance coordinates thus work for infinitesimal boundary layers around \emph{any} smooth surface.}
    \label{fig:sdf-singularities}
\end{figure}

We begin by rescaling the signed-distance coordinate by the small positive number $\eps \ll 1$,
\begin{align}
\sigma &= \eps\, \xi.
\end{align}
Provided that the rescaled normal coordinate $\xi$ is sufficiently small, we can rewrite the geometric factors $J$ and $J^{-1}$ in the form
	\begin{align}
	J &= 1 - \sigma K = 1 - \eps \xi K, &
	J^{-1} &= (1-\sigma K)^{-1} = \sum_{k=0}^\infty \eps^k \xi^k K^{k}.
	\end{align}
The tangential derivative is unaffected by this transformation, whereas the normal derivative gains a power of $\eps^{-1}$, meaning that the gradient and time derivative operators are now given by
	\begin{align}
    \nabla 	&= \eps^{-1} \nhat \partial_\xi + \sum\nolimits_{k=0}^\infty \eps^k \xi^k K^k \sgrad,\\
    \partial_t &= -\eps^{-1} v_\sigma \partial_\xi + (\partial_\tau -v_\bot \cdot \sgrad) + \eps \left(\xi \sgrad v_\sigma \cdot\sum\nolimits_{k=0}^\infty\eps^k\xi^k K^k\sgrad\right).
	\end{align}
We can then rewrite all previous differential operators as formal series in $\eps$	
	\begin{align}
    \nabla f	&= \eps^{-1} \nhat \partial_\xi f + \sum\nolimits_{k=0}^\infty \eps^k \xi^k K^k \sgrad f,\\
    \nabla v &= \eps^{-1}\nhat \otimes (\nhat \,\partial_\xi u_\sigma +  \partial_\xi u_\bot  )\nonumber\\
    &+ \left(\sum\nolimits_{k=0}^\infty \eps^k \xi^k K^{k}\right)\left((\sgrad u_\sigma+ K u_\bot)\otimes \nhat + (\sgrad u_\bot\cdot\Pi - u_\sigma \, K )\right),\\
    |J|\divof{u}&= \eps^{-1} {\partial_\xi u_\sigma}  -\partial_\xi(\xi \Kbar u_\sigma) + \sdiv{u_\bot} + \eps \br{\partial_\xi(\xi^2 \Kabs u_\sigma) - \xi \sdiv(\Khat u_\bot)},\\
	\lapof f &= \eps^{-2} \partial_\xi^2 f + \eps^{-1}(-\Kbar \partial_\xi f) + \eps^0(-\xi \overline{K^2}\partial_\xi f + \sgrad \cdot \sgrad f) + \O(\eps)\\
    -\nabla\times\nabla\times{u}&= \eps^{-2} \partial_\xi^2 u_\bot + \eps^{-1} \br{- \Kbar \partial_\xi u_\bot - \partial_\xi \sgrad u_\sigma- \nhat \sgrad \cdot (\partial_\xi u_\bot)} + \O(\eps^0)\\
   {u}\cdot \grad f &= \eps^{-1} u_\sigma \partial_\xi f + \sum\nolimits_{k=0}^\infty \eps^k \xi^k K^k u_\bot\cdot \sgrad f\\
    {u} \cdot \grad{u}&=\eps^{-1}\left( u_\sigma \partial_\xi u_\bot + \nhat u_\sigma\partial_\xi u_\sigma\right)\\
    &\quad + \left(\sum\nolimits_{k=0}^{\infty} \eps^k\xi^k u_\bot K^k \right) \cdot \br{(\sgrad u_\sigma+ K u_\bot)\nhat + (\sgrad u_\bot - K u_\sigma)} \nonumber.
    \end{align}
For higher order expansions or other vector calculus operators, see the Mathematica code available at \href{https://github.com/ericwhester/signed-distance-code}{github.com/ericwhester/signed-distance-code}.

\section{Signed-distance around lines}
\label{sec:tubes}
\newcommand{\bhat}[0]{\hat{b}}
Not all boundary layers have codimension one.
Edges or boundaries of a surface, as well as singular curves, each call for a new approach to parameterisation.
It is possible to deal with these situations as limiting cases of surfaces with vanishing thickness.
We will instead approach these problems by starting with the Frenet-Serret frame of a curve \cite{FrenetCourbesDoubleCourbure1852,SerretQuelquesFormulesRelatives1851}, deriving orthogonal coordinates based on the Bishop frame \cite{BishopThereMoreOne1975}, and generalising to provide connection coefficients in the neighbourhood of evolving lines.

\subsection{The Frenet-Serret frame}
\label{sec:frenet}
We begin with a smooth parameterised curve $p(s,\tau)$ immersed in space, which evolves smoothly over time $\tau$
	\begin{align}
	p(s,\tau).
	\end{align}
For convenience, we may omit the functional dependence of various quantities on $s,\tau$.

We now seek a frame of vectors along the curve.
We begin with the \emph{tangent vector} to the curve
	\begin{align}
	\label{eq:tangent-tube}
	t = \pd{p(s,\tau)}{s},
	\end{align}
and define the \emph{unit tangent vector} $\that$ to be normalised to unit length
	\begin{align}
	\label{eq:t-tube}
	\hat{t} = \frac{t}{|t|}.
	\end{align}
For convenience we also define the rescaled derivative operator $\nabla_s$ that calculates the arclength derivative
	\begin{align}
	\label{eq:ds}
	\nabla_s = \frac{1}{|t|} \pd{}{s}.
	\end{align}
Because the tangent vector is normalised, its derivative must be perpendicular to $\that$.
This leads to the definition of the \emph{normal vector} $\nhat$ and the \emph{curvature} $\kappa$ 
	\begin{align}
	\label{eq:n-tube}
	\nabla_s \hat{t} &= \kappa(s,\tau) \hat{n}.
	\end{align}

\begin{remark}
If the curve has zero derivative, we can freely pick a normal vector and set $\kappa(s,\tau)=0$.
\end{remark}
We complete our frame with the \emph{binormal vector} $\bhat$, defined to be orthogonal to both $\that$ and $\nhat$
	\begin{align}
	\label{eq:b-tube}
	\bhat &= \that \times \nhat.
	\end{align}
The derivative of the normal vector then defines the \emph{torsion} $\omega$
	\begin{align}
	\label{eq:dsn-tube}
	\nabla_s \nhat &= -\kappa \that + \omega \bhat,
	\end{align}
and we use orthonormality to see that the derivative of the binormal is proportional to the normal vector,
	\begin{align}
	\label{eq:dsb-tube}
	\nabla_s \bhat &= -\omega \nhat.
	\end{align}

\subsection{Orthogonal signed-distance coordinates}
\label{sec:sdf-coords-tube}
A simple coordinate definition near the curve would define an angular coordinate relative to the Frenet-Serret frame.
However we show that in order to preserve an orthogonal coordinate system, we must rotate our angular coordinate according to the torsion of the curve.
\begin{definition}
    We define a new coordinate system in the neighbourhood of the curve using the signed-distance $\sigma$ and an angular coordinate $\theta$
	\begin{align}
	x(s,\theta,\sigma,\tau) &= p(s,\tau) + \sigma \br{\cos\br{\theta+\phi(s,\tau)} \nhat(s,\tau) + \sin\br{\theta+\phi(s,\tau)}\bhat(s,\tau)}.
	\end{align}
\end{definition}
\begin{lemma}
	This coordinate system is orthogonal if 
	\begin{align}
	\nabla_s \phi &= -\omega(s).
	\end{align}	
\end{lemma}
\begin{proof}
The rows of the Jacobian of the coordinate transformation are orthogonal under this condition,
	\begin{align*}
	\begin{bmatrix}
		\nabla_s x \\ \partial_\theta x \\ \partial_\sigma x
	\end{bmatrix} &=
	\begin{bmatrix}
	1 - \sigma (\cos(\theta + \phi)\kappa ) & -\sigma \sin(\theta + \phi) (\omega + \nabla_s \phi) & \sigma \cos(\theta+\phi)(\omega + \nabla_s \phi)\\
	0 & -\sigma \sin(\theta+\phi) & \sigma \cos(\theta+\phi)\\
	0 & \cos(\theta + \phi) & \sin (\theta + \phi)
	\end{bmatrix}.
	\end{align*}
\end{proof}
\begin{remark}
	This condition is equivalent to Bishop's for a relatively parallel adapted frame of a curve \cite{BishopThereMoreOne1975}.
	The frame is not unique, with one degree of freedom for the initial angular coordinate $\phi(0)$.
\end{remark}
\begin{definition}
	We define the coordinate frame unit vectors $(\that_s,\that_\theta,\that_\sigma)$
	\begin{align}
	\label{eq:sdf-frame-tube}
	\that_s &= \frac{\nabla_s x}{h_s} = \that,\\
	\that_\theta &= \frac{\partial_\theta x}{\sigma} = -\sin(\theta+\phi) \nhat + \cos(\theta+\phi) \bhat,\\
	\that_\sigma &= \frac{\partial_\sigma x}{h_\sigma} = \cos(\theta + \phi)\nhat + \sin(\theta+\phi)\bhat,
	\end{align}
where we define the scaling factors
	\begin{align}
	\label{eq:scales-factors-tube}
	h_s &= |\nabla_s x| = 1 - \sigma\kappa \cos(\theta+\phi),&
	\sigma &= |\partial_\theta x| = \sigma,&
	h_\sigma &= |\partial_\sigma x| = 1.
	\end{align}
\end{definition}
With this orthogonal coordinate system, we derive the gradient
\begin{corollary}
The gradient is 
	\begin{align}
	\nabla =  \hat{t}_{\sigma } \partial_\sigma + \frac{1}{h_s}  \hat{t}_s \nabla_s \, + \frac{1}{\sigma} \hat{t}_{\theta }  \partial_\theta.
	\end{align}
\end{corollary}
\begin{proof}
    This follows from the Jacobian of the coordinate transformation and our normalisation of the orthogonal vector frame.
\end{proof}
\begin{lemma}
The connection coefficients of the $(\that_s,\that_\theta,\that_\sigma)$ frame are 	
	\begin{align}
	\label{eq:connection-tube}
	\nabla \hat{t}_s &= 
		\A\, \that_s \otimes \that_\sigma
		+\B\, \that_s\otimes\that_\theta,\\
	\nabla \hat{t}_\sigma &= 
		-\A \,\hat{t}_s   \otimes \hat{t}_s
		+\C\, \hat{t}_\theta \otimes \hat{t}_\theta,\\
	\nabla \hat{t}_\theta &= 
		-\B\,	\hat{t}_s \otimes \hat{t}_s
		-\C \, \hat{t}_\theta \otimes \hat{t}_\sigma,
	\end{align}
where we define
	\begin{align}
	\A &= \frac{1}{h_s} \kappa  \cos (\phi +\theta ), &
	\B &= -\frac{1}{h_s} \kappa  \sin (\phi +\theta ), &
	\C &= \frac{1}{\sigma}.
	\end{align}

\end{lemma}
\begin{proof}
	These follow from the definition of the coordinate frame in terms of the Frenet-Serret frame and the known derivatives of the latter.
\end{proof}

\subsection{Vector calculus operators}
\label{sec:vector-calculus-tube}
We derive the remaining vector calculus operators from the gradient and connection coefficients.
\begin{corollary}
The vector gradient is 
    \begin{align}
    \label{eq:grad-tube}
    \nabla u &= 
    \left(\frac{\nabla_s u_s}{h_s}-\A u_{\sigma }-\B u_{\theta }\right)\hat{t}_s\otimes\hat{t}_s
    + \left(\frac{\nabla_s u_{\sigma }}{h_s}+\A u_s\right)\hat{t}_s\otimes\hat{t}_{\sigma }
    + \left(\frac{\nabla_s u_{\theta }}{h_s}+\B u_s\right)\hat{t}_s\otimes\hat{t}_{\theta }\nonumber\\
        &\quad
    + \left(\partial_{\sigma } u_s\right)\hat{t}_{\sigma }\otimes\hat{t}_s
    + \left(\partial_{\sigma } u_{\sigma }\right)\hat{t}_{\sigma }\otimes\hat{t}_{\sigma }
    + \left(\partial_{\sigma } u_{\theta }\right)\hat{t}_{\sigma }\otimes\hat{t}_{\theta }\nonumber\\
        &\quad
    + \left(\frac{\partial_{\theta } u_s}{\sigma}\right)\hat{t}_{\theta }\otimes\hat{t}_s
    + \left(\frac{\partial_{\theta } u_{\sigma }}{\sigma}-\C u_{\theta }\right)\hat{t}_{\theta }\otimes\hat{t}_{\sigma }
    + \left(\frac{\partial_{\theta } u_{\theta }}{\sigma}+\C u_{\sigma }\right)\hat{t}_{\theta }\otimes\hat{t}_{\theta }.
    \end{align}
\end{corollary}
\begin{corollary}
	The divergence is 
	\begin{align}
	\label{eq:div-tube}
	\nabla \cdot u &= \frac{\nabla_s u_s}{h_s}
    +\partial_{\sigma } u_{\sigma }
    +\frac{\partial_{\theta } u_{\theta }}{\sigma }
    +\left(\C-\A\right)u_{\sigma }
    -\B u_{\theta } .
	\end{align}
\end{corollary}
\begin{proof}
We take the trace of the vector gradient.
\end{proof}
\begin{corollary}
The scalar Laplacian is 
    \begin{align}
    \label{eq:lapf-tube}
    \Delta f &= 
        \frac{\nabla_s^2 f}{h_s^2}
        +\partial_{\sigma }^2 f
        +\frac{\partial_{\theta }^2 f}{\sigma ^2}
        -\frac{ \nabla_s h_s}{h_s^3}\nabla_s f
        +\left(\C-\A\right) \partial_{\sigma } f
        -\frac{\B}{\sigma}\partial_{\theta } f.
    \end{align}
\end{corollary}
\begin{proof}
We take the divergence of the scalar gradient.
\end{proof}
\begin{corollary}
The vector curl is 
	\begin{align}
	\label{eq:curl-tube}
	\nabla \times u &= \left(\partial_{\sigma } u_{\theta }-\frac{\partial_{\theta } u_{\sigma }}{\sigma }+\C{u_{\theta }}\right)\hat{t}_s
    + \left(\frac{\partial_{\theta } u_s}{\sigma }-\frac{\nabla_s u_{\theta }}{h_s}- \B u_s \right)\hat{t}_{\sigma }
    + \left(\frac{\nabla_s u_{\sigma }}{h_s}-\partial_{\sigma } u_s+\A u_s \right)\hat{t}_{\theta }.
	\end{align}
\end{corollary}
\begin{proof}
We contract the gradient with the Levi-Civita tensor $(\nabla\times u)\cdot\that_i \equiv \eps_{i,j,k}(\nabla u)_{j,k}$.
\end{proof}
\begin{corollary}
The vector Laplacian is 

\begin{align}
    \Delta u &=
    \left(\Delta u_s
    -\frac{2 \A }{h_s}\nabla_s u_{\sigma }
    -\frac{2 \B }{h_s}\nabla_s u_{\theta }
    +\left(-\A^2-\B^2\right) u_s
    -\frac{ \nabla_s \A}{h_s}u_{\sigma }
    -\frac{\nabla_s \B}{h_s}u_{\theta } \right) \that_s \nonumber\\
    &
    +\left(\Delta u_\sigma 
    +\frac{2 \A }{h_s}\nabla_s u_s
    -\frac{2 \C }{\sigma}\partial_{\theta }u_{\theta }
    +\frac{\nabla_s \A}{h_s}u_s 
    +\left(-\A^2-\C^2\right) u_{\sigma }
    +(-\A \B+\B \C) u_{\theta }\right) \that_\sigma \nonumber\\
    &
    +\left(\Delta u_\theta
    +\frac{2 \B }{h_s}\nabla_s u_s
    +\frac{2 \C }{\sigma}\partial_{\theta } u_{\sigma }
    +\frac{\nabla_s \B}{h_s}u_s 
    +(-\A \B-\B \C) u_{\sigma }
    +\left(-\B^2-\C^2\right) u_{\theta }\right)\that_\theta.
\end{align}

\end{corollary}

\begin{figure}[hbt]
	\centering
    \includegraphics[width=\linewidth]{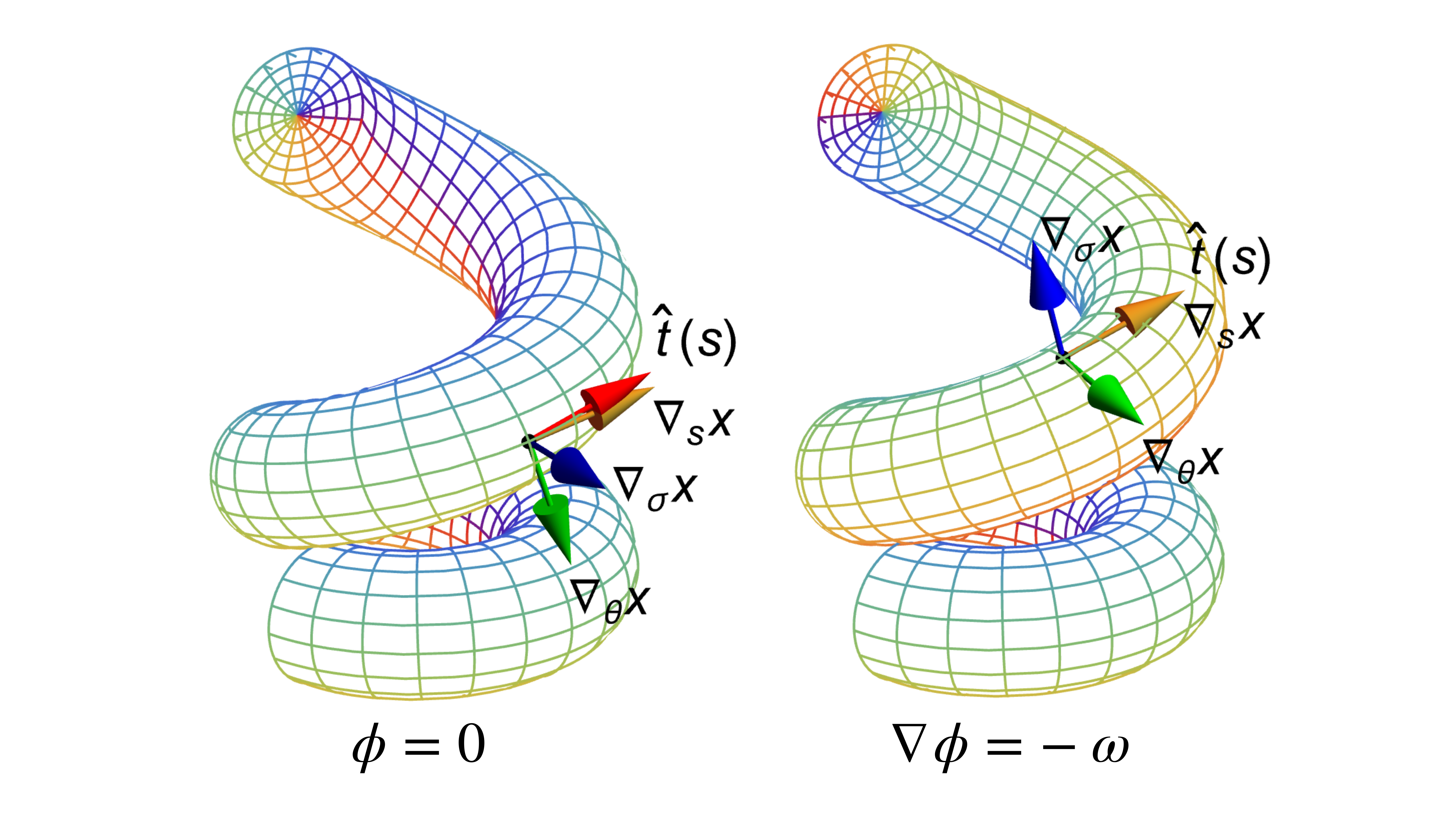}
    \caption{
    A comparison of coordinate systems around the curve $p(s,t) = (\cos 2\pi s, \sin(2\pi s),s^2)$. 
    Left: Coordinates derived from the Frenet-Serret frame, with angular coordinates relative to the normal vector. 
    When the torsion is nonzero, the coordinate tangent vectors are no longer orthogonal away from the curve (compare $\that(s)$ (red) with $\nabla_s x$ (orange)).
    Right: Orthogonal signed-distance coordinates obtained by rotating the angular coordinate according to the torsion $\nabla_s \phi = - \omega$.
    Level sets of $s$ and $\theta$ are shown for both systems at $\sigma = 0.5$ and are each coloured by $\theta$.}
\end{figure}

\subsection{Evolving curves}
\label{sec:evolving-tubes}
In general, the curve may also time evolve.
We want to derive the time derivative operator and its effect on the coordinate frame.
\begin{definition}
The velocity of the curve is defined in terms of the Frenet-Serret frame
	\begin{align}
	\label{eq:dtp-tube}
	\partial_\tau p(\tau,s) = v_t \, \hat{t} + v_n \, \hat{n} + v_b \, \hat{b} \equiv v.
	\end{align}
\end{definition}
\begin{lemma}
The derivatives of the Frenet-Serret frame are 
	\begin{align}
	\label{eq:dt-frenet}
	\partial_\tau \that &= \left(\nabla_s v_t-\kappa  v_n\right)\hat{t}
 + \left(-\omega  v_b+\nabla_s v_n+\kappa  v_t\right)\hat{n}
 + \left(\nabla_s v_b+\omega  v_n\right)\hat{b},\\
 	\partial_\tau \nhat &= 
 	\left(\omega  v_b-\frac{v_n \nabla_s \kappa }{\kappa }-2 \nabla_s v_n+\frac{\nabla_s^2 v_t}{\kappa }-\kappa  v_t\right)\hat{t}\nonumber\\
 	&\quad
 + \frac{1}{\kappa}\left(-{2 \omega  \nabla_s v_b}-{v_b \nabla_s \omega }-{\partial_{\tau } \kappa }+{\nabla_s^2 v_n}+{v_t \nabla_s \kappa }+2 \kappa \nabla_s v_t-{\omega ^2 v_n}-\kappa^2  v_n\right)\hat{n}\nonumber\\
 	&\quad
 + \frac{1}{\kappa}\left({\nabla_s^2 v_b}-{\omega ^2 v_b}+{2 \omega  \nabla_s v_n}+{v_n \nabla_s \omega }+\kappa \omega  v_t\right)\hat{b},\\
 	\partial_\tau \bhat &= 
 	\left(-\nabla_s v_b-\omega  v_n\right)\hat{t}\nonumber\\
 	&\quad
 + \frac{1}{\kappa}\left(-{\nabla_s^2 v_b}+{\omega ^2 v_b}-{2 \omega  \nabla_s v_n}-{v_n \nabla_s \omega }-\kappa \omega  v_t\right)\hat{n}\nonumber\\
 	&\quad
 + \frac{1}{\kappa}\left(-{2 \omega  \nabla_s v_b}-{v_b \nabla_s \omega }-{\partial_{\tau } \kappa }+{\nabla_s^2 v_n}+{v_t \nabla_s \kappa }+3 \kappa \nabla_s v_t-{\omega ^2 v_n}-2 \kappa^2  v_n\right)\hat{b}.
	\end{align}
\end{lemma}
\begin{proof}
These follow from writing each frame vector in terms of succeeding derivatives of the position vector $p(s,\tau)$ (\cref{eq:b-tube,eq:n-tube,eq:t-tube}), applying the time derivative operator $\partial_\tau$, commuting partial derivatives, and converting back to the Frenet-Serret frame.
\end{proof}
\begin{corollary}
The time derivative of the Frenet-Serret frame implies the constraints
	\begin{align}
	\label{eq:dt-constaints-tube}
    \nabla_s v_t-\kappa  v_n &= 0,\\
    -v_b \nabla_s \omega -2 \omega  \nabla_s v_b-\partial_{\tau } \kappa +\nabla_s^2 v_n+v_t \nabla_s \kappa +\left(\kappa ^2-\omega ^2\right) v_n &= 0.
	\end{align}
\end{corollary}
\begin{proof}
Orthonormality of the Frenet-Serret frame implies
$\that\cdot\partial_\tau \that = \tfrac{1}{2}\partial_\tau(\that\cdot \that) = 0$, and similarly for the remaining vectors.
\end{proof}
\begin{corollary}
The Frenet-Serret basis time derivatives take the form
	\begin{align}
	\label{eq:dtau-frenet}
	\partial_\tau \begin{bmatrix}
		\that\\
		\nhat\\
		\bhat
	\end{bmatrix} &=
	\begin{bmatrix}
		0 & \alpha & \beta\\
		-\alpha & 0 & \gamma\\
		-\beta & -\gamma & 0
	\end{bmatrix}
	\begin{bmatrix}
		\that\\
		\nhat\\
		\bhat
	\end{bmatrix},
	\end{align}
where we define
	\begin{align*}
	\alpha &= \kappa  v_t -\omega  v_b+\nabla_s v_n,&
	\beta &= \omega  v_n+ \nabla_s v_b ,&
	\gamma &= \frac{\omega \alpha +  \nabla_s\beta}{\kappa }.
	\end{align*}
\end{corollary}
\begin{lemma}
The partial time derivatives of the signed-distance frame are 
\begin{align}
\label{eq:dtau-sdf-tube}	
\partial_\tau \begin{bmatrix}
	\that_s\\
	\that_\sigma\\
	\that_\theta
\end{bmatrix} &=
\begin{bmatrix}
	0 & \alpha' & \beta'\\
	-\alpha' & 0 & \gamma'\\
	-\beta' & -\gamma' & 0
\end{bmatrix}
\begin{bmatrix}
	\that_s\\
	\that_\sigma\\
	\that_\theta
\end{bmatrix},
\end{align}
where we have defined
\begin{align*}
\alpha' &= \cos(\theta+\phi) \alpha + \sin(\theta+\phi) \beta&
\beta' &= -\sin(\theta+\phi) \alpha + \cos(\theta+\phi) \beta&
\gamma' &= \gamma + \partial_\tau \phi.
\end{align*}
\end{lemma}
\begin{proof}
Write the signed-distance frame in terms of the Frenet-Serret frame, differentiate by applying the above rule, and transform it back to the signed-distance frame.
\end{proof}
Equipped with the moving time derivatives of the basis vectors, we can then calculate the Cartesian time derivatives of the coordinates, defined by 
 	\begin{align}
 	\label{eq:sdf-tube-t}
	t &= \tau, &
	x &= p(s,\tau) + \sigma \that_\sigma(s,\theta,\tau).
	\end{align}
\begin{proposition}
The Cartesian time derivative operator is 
	\begin{align}
	\label{eq:dt-tube}
	\partial_t = \partial_\tau  + \partial_t s \partial_s + \partial_t \sigma \partial_\sigma + \partial_t \theta \partial_\theta.
	\end{align}
\end{proposition}

\begin{lemma}
The Cartesian time derivatives of signed-distance coordinates (holding $x$ constant) are 
	\begin{align}
	\label{eq:dt-coords-tube}
	\partial_t s &= -\frac{1}{|t| h_s}\br{v_t - \sigma \alpha'}, & 
	\partial_t \sigma &= - v_\sigma, &
	\partial_t \theta &= -\frac{1}{\sigma}(v_\theta + \sigma \gamma'),
	\end{align}
where we have defined
    \begin{align}
    v_\sigma &= \that_\sigma \cdot v = \cos(\theta+\phi) v_n + \sin(\theta+\phi) v_b, & 
    v_\theta &= \that_\theta \cdot v = - \sin(\theta+\phi) v_n+ \cos(\theta+\phi) v_b.
    \end{align}
\end{lemma}
\begin{proof}
Apply the defined Cartesian time derivative (\cref{eq:dt-tube}) to the signed-distance function transform of the Cartesian spatial coordinates (\cref{eq:sdf-tube-t}), and solve for the time derivatives.
\end{proof}

\begin{corollary}
The Cartesian time derivatives of the signed-distance basis vectors are	
	\begin{align}
	\label{eq:dt-basis-tube}	
	\partial_t \begin{bmatrix}
		\that_s\\
		\that_\sigma\\
		\that_\theta
	\end{bmatrix} &= 
	\begin{bmatrix}
		0 & a & b\\
            -a & 0 & c\\
            -b & -c & 0
	\end{bmatrix}
	\begin{bmatrix}
		\that_s\\
		\that_\sigma\\
		\that_\theta
	\end{bmatrix},
	\end{align}
where we define
    \begin{align}
        a &= \frac{\alpha'- v_t \A}{h_s}, &
        b &= \beta'+\frac{\B(\sigma \alpha' -v_t)}{h_s}, &
        c &= -\frac{1}{\sigma} v_\theta.
    \end{align}
\end{corollary}

    \begin{corollary}
	The Cartesian time derivative of a scalar is 
	\begin{align}
	\partial_t f &= \partial_{\tau } f - \frac{\left(v_t-\sigma \alpha' \right)}{h_s}\nabla_s f -v_\sigma \partial_{\sigma } f
-\frac{(v_\theta + \sigma \gamma')}{\sigma }\partial_{\theta } f.
	\end{align}
    \end{corollary}
    \begin{corollary}
    The Cartesian time derivative of a vector is 
	\begin{align}
	\partial_t u &=
     \left(
     \partial_t u_s
     -a u_{\sigma }
     -b u_{\theta }\right)\hat{t}_s
     +\left(
     \partial_t u_\sigma
     +a u_s-c u_{\theta }\right)\hat{t}_{\sigma }
     + \left(
     \partial_t u_\theta
     +b u_s+c u_{\sigma }\right)\hat{t}_{\theta }.
    \end{align}
\end{corollary}

    \begin{corollary}
    The time derivative of the torsion satisfies
    \begin{align}
    \partial_{\tau } \omega  &=
    \frac{\omega ^2 v_b \nabla_s \kappa }{\kappa ^2}
    -\frac{\nabla_s \kappa  \nabla_s^2 v_b}{\kappa ^2}
    -\frac{\omega ^2 \nabla_s v_b}{\kappa }
    -\frac{2 \omega  v_b \nabla_s \omega }{\kappa }
    +\kappa  \nabla_s v_b
    +\frac{\nabla_s^3 v_b}{\kappa }
    -\frac{2 \omega  \nabla_s \kappa  \nabla_s v_n}{\kappa ^2}\nonumber\\
    &\quad
    -\frac{v_n \nabla_s \kappa  \nabla_s \omega }{\kappa ^2}
    +\frac{2 \omega  \nabla_s^2 v_n}{\kappa }
    +\frac{3 \nabla_s \omega  \nabla_s v_n}{\kappa }
    +\frac{v_n \nabla_s^2 \omega }{\kappa }
    +v_t \nabla_s \omega 
    +2 \kappa  \omega  v_n
    \end{align}
    \end{corollary}
    \begin{proof}
    This follows from enforcing that the commutator $\partial_t\nabla \theta-\nabla\partial_t \theta$ is identically zero.
    \end{proof}
\subsection{Application: Curve Boundary Layers}
\label{sec:asymptotics-tube}
This geometric machinery allows us to describe all vector calculus near moving curves.
Similar to the surface case, we rescale our signed-distance coordinate by $\eps$
	\begin{align}
    \sigma &= \eps \xi, &
    h_s^{-1} &= \left(\sum_{k=0}^\infty (\eps \xi \cs \kappa)^k \nabla_s \right).
	\end{align}
\begin{remark}
    For convenience in the remainder of the section we will abbreviate
    \begin{align}
    \cos(\theta+\phi)&= \cs, &
    \sin(\theta+\phi)&= \sn.
    \end{align}
\end{remark}
From this rescaling, the spatial and temporal derivatives are rescaled by
    \begin{align}
    \nabla &= \hat{t}_s\left(\sum_{k=0}^\infty (\eps \xi \kappa \cs)^k \right)\nabla_s + \frac{\hat{t}_{\sigma }}{\varepsilon }\partial_{\xi } + \frac{\hat{t}_{\theta }}{\varepsilon  \xi }\partial_{\theta } ,\\
    \partial_t &= \partial_{\tau } 
        -\left(\sum_{k=0}^\infty(\eps \xi \kappa \cs)^k\right){\left(v_t-\varepsilon  \xi \alpha'\right)}\nabla_s 
	-\frac{v_\sigma   }{\varepsilon }\partial_{\xi } 
	-\frac{\left(v_\theta+\varepsilon  \xi \gamma' \right) }{\varepsilon  \xi }\partial_{\theta } 
	\end{align}
The remaining operators result from applying or contracting these operators, using connection coefficients to determine values for higher rank tensors, and solving order-by-order in $\eps$.
    \begin{align}
    \nabla f &= \left(\sum_{k=0}^\infty (\eps \xi \kappa \cs)^k \nabla_s f\right)\hat{t}_s
        + \frac{\partial_{\xi } f}{\varepsilon }\hat{t}_{\sigma }
        + \frac{\partial_{\theta } f}{\varepsilon  \xi }\hat{t}_{\theta },\\
    \nabla u &= \left(\sum_{k=0}^\infty(\eps \xi \kappa c)^k\right)\hat{t}_s\otimes\left(
        \left(\nabla_s u_s-{\kappa  \left(c u_{\sigma }-s u_{\theta }\right)}\right)\hat{t}_s
        + \left({c \kappa  u_s+\nabla_s u_{\sigma }}\right)\hat{t}_{\sigma }
        - \left(\kappa  s u_s+\nabla_s u_{\theta }\right)\hat{t}_{\theta }\right)\nonumber\\
        &\quad
        + \left(\frac{\partial_{\xi } u_s}{\varepsilon }\right)\hat{t}_{\sigma }\otimes\hat{t}_s
        + \left(\frac{\partial_{\xi } u_{\sigma }}{\varepsilon }\right)\hat{t}_{\sigma }\otimes\hat{t}_{\sigma }
        + \left(\frac{\partial_{\xi } u_{\theta }}{\varepsilon }\right)\hat{t}_{\sigma }\otimes\hat{t}_{\theta }\nonumber\\
        &\quad
        + \left(\frac{\partial_{\theta } u_s}{\varepsilon  \xi }\right)\hat{t}_{\theta }\otimes\hat{t}_s
        + \left(\frac{\partial_{\theta } u_{\sigma }-u_{\theta }}{\varepsilon  \xi }\right)\hat{t}_{\theta }\otimes\hat{t}_{\sigma }
        + \left(\frac{\partial_{\theta } u_{\theta }+u_{\sigma }}{\varepsilon  \xi }\right)\hat{t}_{\theta }\otimes\hat{t}_{\theta },\\
    \nabla \cdot u &= 
        \left(\sum_{k=0}^\infty(\eps \xi \kappa c)^k\right) (\nabla_s u_s-\kappa  \left(c u_{\sigma }-s u_{\theta }\right)) + \frac{\partial_\xi u_\sigma}{\eps}  + \frac{u_\sigma+\partial_\theta u_\theta}{\eps\xi},\\
    \Delta f &= 
        \frac{\frac{\partial_{\theta }^2 f}{\xi ^2}
        +\frac{\partial_{\xi } f}{\xi }+\partial_{\xi }^2 f}{\varepsilon ^2}
        +\frac{\frac{\kappa  \sn \partial_{\theta } f}{\xi }-\cs \kappa  \partial_{\xi } f}{\varepsilon }
        +\left(\nabla_s^2 f-\cs^2 \kappa ^2 \xi  \partial_{\xi } f+\cs \kappa ^2 \sn \partial_{\theta } f\right)+O\left(\varepsilon\right),\\
       \Delta u&=
        \left(\Delta u_s -\kappa  \left(2 \cs \nabla_s u_{\sigma }-2 \sn \nabla_s u_{\theta }+\omega  \left(\cs u_{\theta }+\sn u_{\sigma }\right)\right)+\nabla_s \kappa  \left(\sn u_{\theta }-\cs u_{\sigma }\right)+\kappa ^2 \left(-u_s\right)\right)\hat{t}_s\nonumber\\
        &+ 
        \left(\Delta u_\sigma-\frac{2 \partial_{\theta } u_{\theta }+u_{\sigma }}{\varepsilon ^2 \xi ^2}-\frac{\kappa  \sn u_{\theta }}{\varepsilon  \xi }+\left(-\cs^2 \kappa ^2 u_{\sigma }+\kappa  \left(2 \cs \nabla_s u_s+\sn \omega  u_s\right)+\cs u_s \nabla_s \kappa \right)\right)\hat{t}_{\sigma }\nonumber\\
        &+ 
        \left(\Delta u_\theta + \frac{2 \partial_{\theta } u_{\sigma }-u_{\theta }}{\varepsilon ^2 \xi ^2}+\frac{\kappa  \sn u_{\sigma }}{\varepsilon  \xi }\right.\nonumber\\
        & \quad +\left.\left(\kappa  \left(\cs \omega  u_s-2 \sn \nabla_s u_s\right)-\sn u_s \nabla_s \kappa +\kappa ^2 \left(2\sn\cs u_{\sigma } -\sn^2 u_{\theta }\right)\right)\right)\hat{t}_{\theta } +O\left(\varepsilon\right),\\
        \nabla \times u &=
        \left(\frac{\xi  \partial_{\xi } u_{\theta }-\partial_{\theta } u_{\sigma }+u_{\theta }}{\varepsilon  \xi }\right)\hat{t}_s\nonumber\\
        &\quad
        + \left(\frac{\partial_{\theta } u_s}{\varepsilon  \xi }+\left(\kappa  \sn u_s-\nabla_s u_{\theta }\right)+\varepsilon  \left(\cs \kappa ^2 \xi  \sn u_s-\cs \kappa  \xi  \nabla_s u_{\theta }\right)\right)\hat{t}_{\sigma }\nonumber\\
        &\quad
        + \left(-\frac{\partial_{\xi } u_s}{\varepsilon }+\left(\cs \kappa  u_s+\nabla_s u_{\sigma }\right)+\varepsilon  \left(\cs^2 \kappa ^2 \xi  u_s+\cs \kappa  \xi  \nabla_s u_{\sigma }\right)\right)\hat{t}_{\theta }+O\left(\varepsilon ^2\right),\\
    \partial_t f &= 
        -\frac{v_\sigma \partial_\xi f + \xi^{-1} v_\theta \partial_\theta f}{\eps} + \partial_{\tau } f
        -\gamma'\partial_{\theta } f
        -\left(\sum_{k=0}^\infty(\eps \xi \kappa c)^k\right){\left(v_t-\varepsilon  \xi  \alpha'\right)}\nabla_s f,\\
    \partial_t u &=
        \left(\left(u_{\sigma } \left(\A' v_t-\alpha'\right)+u_{\theta } \left(\B' v_t-\beta'\right)\right)-\varepsilon  \xi  \left(\A' u_{\sigma }+\B' u_{\theta }\right) \left(\alpha'-\cs \kappa  v_t\right)\right)\hat{t}_s\nonumber\\
        &
        + \left(\frac{u_{\theta } v_{\theta }}{\varepsilon  \xi }+u_s \left(\alpha'-\A' v_t\right)+\A' \varepsilon  \xi  u_s \left(\alpha'-\cs \kappa  v_t\right)\right)\hat{t}_{\sigma }\nonumber\\
        &
        + \left(-\frac{u_{\sigma } v_{\theta }}{\varepsilon  \xi }+u_s \left(\beta'-\B' v_t\right)+\B' \varepsilon  \xi  u_s \left(\alpha'-\cs \kappa  v_t\right)\right)\hat{t}_{\theta }+O\left(\varepsilon ^2\right)
    \end{align}
where we defined $\A' = h_s \A=\cos(\theta+\phi)\kappa, \B'=h_s \B=-\sin(\theta+\phi)\kappa$.

\section{Conclusion}
\label{sec:conclusion}
We provide an elementary derivation of an orthogonal coordinate system for neighbourhoods of evolving smooth surfaces and curves based on the signed-distance function.
We go beyond previous works on the signed-distance function to provide proofs and a computational algebra framework\footnote{Available at \href{https://github.com/ericwhester/signed-distance-code}{github.com/ericwhester/signed-distance-code}.} for a wide range of useful vector calculus identities.
Our results thereby enable consistent accounting of geometric effects in the derivation of boundary layer asymptotics for a wide range of physical systems.

\section{Acknowledgements}
Eric Hester wishes to acknowledge support from the American Mathematical Society and the Simons Foundations through the AMS-Simons Travel Grant.
We thank Alexander Morozov for helpful comments.

\bibliographystyle{siam}
\bibliography{signed-distance-function.bib}

\end{document}